\documentclass[11pt]{article}
\usepackage{tikz-qtree}
\usetikzlibrary{decorations.pathmorphing}
\usepackage{rotating}

\usepackage[T1]{fontenc}
\usepackage[utf8]{inputenc}
\usepackage{geometry}
\usepackage{color}
\usepackage{float}
\usepackage{amsmath}
\usepackage{amsthm}
\usepackage{amssymb}
\usepackage[noend]{algorithmic}
\usepackage[boxed]{algorithm}
\usepackage[unicode=true,pdfusetitle,
 bookmarks=true,bookmarksnumbered=false,bookmarksopen=false,
 breaklinks=false,pdfborder={0 0 1},colorlinks=false]
 {hyperref}
\usepackage{breakurl}
\usepackage{xspace}

\usepackage{subcaption}

\theoremstyle{plain}
\newtheorem{thm}{Theorem}
  \theoremstyle{remark}
  
  \theoremstyle{definition}
  
  \theoremstyle{plain}
  \newtheorem{lemma}{Lemma}
  \theoremstyle{remark}
  
  \theoremstyle{cor}

\renewcommand{\NG}[1]{{\rm NG}(#1)}
\newcommand{\LG}[1]{{\rm LG}(#1)}
\newcommand{\CNG}[2][]{{\rm CNG}_{#1}(#2)}

\newcommand{\CLG}[2][]{{\rm CLG}_{#1}(#2)}
\newcommand{\GG}[2][k]{{\rm GG}_{#1}(#2)}

\newcommand{\udot}[1]{%
    \tikz[baseline=(todotted.base)]{
        \node[inner sep=1pt,outer sep=0pt] (todotted) {#1};
        \draw[dotted, thick] (todotted.south west) -- (todotted.south east);
    }%
}%
\newcommand{\udash}[1]{%
    \tikz[baseline=(todotted.base)]{
        \node[inner sep=1pt,outer sep=0pt] (todotted) {#1};
        \draw[dashed] (todotted.south west) -- (todotted.south east);
    }%
}%
\newcommand{\uline}[1]{%
    \tikz[baseline=(todotted.base)]{
        \node[inner sep=1pt,outer sep=0pt] (todotted) {#1};
        \draw (todotted.south west) -- (todotted.south east);
    }%
}%
\newcommand{\red}[1]{{\color{red}\uline{$#1$}}}
\newcommand{\blue}[1]{{\color{blue}\udash{$#1$}}}
\newcommand{\purple}[1]{{\color{purple}\udot{$#1$}}}

\title{Clique-Based Lower Bounds for \\ Parsing Tree-Adjoining Grammars}

\author{Karl Bringmann\footnote{Max Planck Institute for Informatics, Saarland Informatics Campus, Saarbr\"ucken, Germany,
\texttt{kbringma@mpi-inf.mpg.de}}
\and 
Philip Wellnitz\footnote{Saarland University, Saarland Informatics Campus, Saarbr\"ucken, Germany, \texttt{s8phwell@stud.uni-saarland.de}}}

\date{}

\begin{document}
\maketitle

\begin{abstract}
    Tree-adjoining grammars are a generalization of context-free grammars that are well suited to model human languages and are thus popular in computational linguistics.
    In the tree-adjoining grammar recognition problem, given a grammar $\Gamma$ and a string $s$ of length $n$, the task is to decide whether $s$ can be obtained from $\Gamma$.
    Rajasekaran and Yooseph's parser (JCSS'98) solves this problem in time $O(n^{2\omega})$, where $\omega < 2.373$ is the matrix multiplication exponent. The best algorithms avoiding fast matrix multiplication take time $O(n^6)$.

    The first evidence for hardness was given by Satta (J.~Comp.~Linguist.'94): For a more general parsing problem, any algorithm that avoids fast matrix multiplication and is significantly faster than $O(|\Gamma|\, n^6)$ in the case of $| \Gamma | = \Theta(n^{12})$ would imply a breakthrough for Boolean matrix multiplication.

    Following an approach by Abboud et al.~(FOCS'15) for context-free grammar recognition, in this paper we resolve many of the disadvantages of the previous lower bound. We show that, even on constant-size grammars, any improvement on Rajasekaran and Yooseph's parser would imply a breakthrough for the $k$-Clique problem.
    This establishes tree-adjoining grammar parsing as a practically relevant problem with the unusual running time of $n^{2 \omega}$, up to lower order factors.
\end{abstract}

\section{Introduction}

Introduced in~\cite{joshi1985tree,joshi1975tree}, tree-adjoining grammars (TAGs) are a system to manipulate certain trees to arrive at strings, see Section~\ref{sec:prelim} for a definition.
TAGs are more powerful than context-free grammars, capturing various phenomena of human languages which require more formal power; in particular TAGs have an ``extended domain of locality'' as they allow ``long-distance dependencies''~\cite{joshi1997tree}.
These properties, and the fact that TAGs are efficiently parsable~\cite{VSJ85}, make them highly desirable in the field of computer linguistics. This is illustrated by the large literature on variants of TAGs (see, e.g.,~\cite{demberg2013incremental,resnik1992probabilistic,shieber1990synchronous,vijay1988feature}), their formal language properties (see, e.g.,~\cite{joshi1997tree,VSJ85}), as well as practical applications (see, e.g.,~\cite{abeille1988parsing,forbes2003d,stone1997sentence,uemura1999tree}), including the XTAG project which developed a tree-adjoining grammar for (a large fraction of) the English language~\cite{doran1994xtag}.
In fact, TAGs are so fundamental to computer linguistics that there is a biannual meeting called ``International Workshop on Tree-Adjoining Grammars and Related Formalisms''~\cite{DBLP:conf/tag/2016}, and they are part of their undergraduate curriculum (at least at Saarland University).

The prime algorithmic problems on TAGs are \emph{parsing} and \emph{recognition}. In the recognition problem, given a TAG $\Gamma$ and a string $s$ of length $n$, the task is to decide whether $\Gamma$ can generate $s$.
The parsing problem is an extended variant where in case $\Gamma$ can generate $s$ we should also output a sequence of derivations generating $s$.
The first TAG parsers ran in time\footnote{In most running time bounds we ignore the dependence on the grammar size, as we are mostly interested in constant-size grammars in this paper.} $O(n^6)$~\cite{schabes1988earley,VSJ85}, which was improved by
Rajasekaran and Yooseph~\cite{rajasekaran1998tal} to $O(n^{2 \omega})$, where $\omega < 2.373$ is the exponent of (Boolean) matrix multiplication.

A limited explanation for the complexity of TAG parsing was given by Satta~\cite{S94}, who designed a reduction from Boolean matrix multiplication to TAG parsing, showing that any TAG parser running faster than $O(|\Gamma| \, n^6)$ on grammars of size $|\Gamma| = \Theta(n^{12})$ yields a Boolean matrix multiplication algorithm running faster than $O(n^3)$. This result has several shortcomings: (1) It holds only for a more general parsing problem, where we need to determine for each substring of the given string $s$ whether it can be generated from $\Gamma$. (2) It gives a matching lower bound only in the unusual case of $|\Gamma| = \Theta(n^{12})$, so that it cannot exclude time, e.g., $O(|\Gamma|^2 n^4)$. (3) It gives matching bounds only restricted to \emph{combinatorial} algorithms, i.e., algorithms that avoid fast matrix multiplication\footnote{The notion of ``combinatorial algorithms'' is informal, intuitively meaning that we forbid unpractical algorithms such as fast matrix multiplication. It is an open research problem to find a reasonable formal definition.}. Thus, so far there is no satisfying explanation of the complexity of TAG parsing.

\paragraph*{Context-free grammars}
The classic problem of parsing \emph{context-free grammars}, which has important applications in programming languages, was in a very similar situation as TAG parsing until very recently. Parsers in time $O(n^3)$ were known since the 60s~\cite{cocke,earley1970efficient,kasami,younger1967recognition}. In a breakthrough, Valiant~\cite{valiant1975general} improved this to $O(n^\omega)$. Finally, a reduction from Boolean matrix multiplication due to Lee~\cite{lee2002fast} showed a matching lower bound for combinatorial algorithms for a more general parsing problem in the case that the grammar size is $\Theta(n^6)$.

Abboud et al.~\cite{ABV15} gave the first satisfying explanation for the complexity of context-free parsing, by designing a reduction from the classic $k$-Clique problem, which asks whether there are $k$ pairwise adjacent vertices in a given graph~$G$. For this problem, for any fixed $k$ the trivial running time of $O(n^k)$ can be improved to $O(n^{\omega k /3})$ for any $k$ divisible by~3~\cite{nevsetvril1985complexity} (see~\cite{eisenbrand2004complexity} for the case of $k$ not divisible by~3). The fastest combinatorial algorithm runs in time $O(n^k / \log^k n)$~\cite{vassilevska2009efficient}. The \emph{$k$-Clique hypothesis} states that both running times are essentially optimal, specifically that $k$-Clique has no $O(n^{(\omega/3 - \varepsilon)k})$ algorithm and no combinatorial $O(n^{(1-\varepsilon)k})$ algorithm for any $k \ge 3, \varepsilon > 0$.
The main result of Abboud et al.~\cite{ABV15} is a reduction from the $k$-Clique problem to context-free grammar recognition on a specific, constant-size grammar $\Gamma$, showing that any $O(n^{\omega - \varepsilon})$ algorithm or any combinatorial $O(n^{3-\varepsilon})$ algorithm for context-free grammar recognition would break the $k$-Clique hypothesis. This matching conditional lower bound removes all disadvantages of Lee's lower bound at the cost of introducing the $k$-Clique hypothesis, see~\cite{ABV15} for further discussions.

\paragraph*{Our contribution}
We extend the approach by Abboud et al.\ to the more complex setting of TAGs. Specifically, we design a reduction from the $6k$-Clique problem to TAG recognition:

\begin{thm}
    There is a tree-adjoining grammar $\Gamma$ of constant size such that if we can decide in time $T(n)$ whether a given string of length $n$ can be generated from $\Gamma$, then $6k$-Clique can be solved in time $O\big(T(n^{k + 1} \log n)\big)$, for any fixed $k \ge 1$.  This reduction is combinatorial.
\end{thm}

Via this reduction, any $O(n^{2\omega-\varepsilon})$ algorithm for TAG recognition would prove that $6k$\nobreakdash-Clique is in time $\tilde O(n^{(2\omega-\varepsilon) (k + 1)}) = O(n^{(\omega/3-\varepsilon/9)6 k})$, for sufficiently large\footnote{For this and the next statement it suffices to set $k > 18/\varepsilon$.} $k$.
Furthermore, any combinatorial $O(n^{6-\varepsilon})$ algorithm for TAG recognition would yield a combinatorial algorithm for $6k$-Clique in time $\tilde O(n^{(6-\varepsilon)(k+1)}) = O(n^{(1-\varepsilon/9) 6k})$, for sufficiently large $k$. As both implications would violate the $6k$-Clique conjecture, we obtain tight conditional lower bounds for TAG recognition. As our result (1) works directly for TAG recognition instead of a more general parsing problem, (2) holds for constant size grammars, and (3) does not need the restriction to combinatorial algorithms, it overcomes all shortcomings of the previous lower bound based on Boolean matrix multiplication, at the cost of using the well-established $k$-Clique hypothesis, which has also been used in~\cite{ABV15,backurs_et_al,backurs2016improving,bringmann2016dichotomy,DBLP:conf/cpm/Chang16}.

We thus establish TAG parsing as a practically relevant problem with the quite unusual running time of $n^{2\omega}$, up to lower order factors.
This is surprising, as the authors are aware of only one other problem with a (conjectured or conditional) optimal running time of $n^{2 \omega \pm o(1)}$, namely 6-Clique.

\paragraph*{Techniques}
The essential difference of tree-adjoining and context-free grammars is that the former can grow strings at four positions, see Figure~\ref{fg:normaltree}. Writing a vertex $v_1$ in one position of the string, and writing the neighborhoods of vertices $v_2,v_3,v_4$ at other positions in the string, a simple tree-adjoining grammar can test whether $v_1$ is adjacent to $v_2,v_3$, and $v_4$. Extending this construction, for $k$-cliques $C_1,C_2,C_3,C_4$ we can test whether $C_1 \cup C_2, C_1 \cup C_3$, and $C_1 \cup C_4$ form $2k$-cliques. Using two permutations of this test, we ensure that $C_1 \cup C_2 \cup C_3 \cup C_4$ forms an almost-$4k$-clique, i.e., only the edges $C_3 \times C_4$ might be missing (in Figure~\ref{fg:claw} below this situation is depicted for cliques $C_2,C_5,C_1,C_6$ instead of $C_1,C_2,C_3,C_4$).
Finally, we use that a $6k$-clique can be decomposed into 3 almost-$4k$-cliques, see Figure~\ref{fg:decomposealmostfourclique}.

In the constructed string we essentially just enumerate 6 times all $k$-cliques of the given graph $G$, as well as their neighborhoods, with appropriate padding symbols (see Section~\ref{sec:string}).
We try to make the constructed tree-adjoining grammar as easily accessible as possible by defining a certain programming language realized by these grammars, and phrasing our grammar in this language, which yields subroutines with an intuitive meaning (see Section~\ref{sec:programming}).

\section{Preliminaries on tree-adjoining grammars} \label{sec:prelim}

In this section we define tree-adjoining grammars and give examples. Fix a set $T$ of terminals and a set $N$ of non-terminals. In the following, conceptually we partition the nodes of any tree into its \emph{leaves}, the \emph{root}, and the remaining \emph{inner nodes}.
An \emph{initial tree} is a rooted tree where
\begin{itemize}
    \item the root and each inner node is labeled with a non-terminal,
    \item each leaf is labeled with a terminal, and
    \item each inner node can be marked for adjunction.
\end{itemize}
See Figure~\ref{fg:initialauxiliary} for an example; nodes marked for adjunction are annotated by a rectangle. An \emph{auxiliary tree} is a rooted tree where
\begin{itemize}
    \item the root and each inner node is labeled with a non-terminal,
    \item \emph{exactly one leaf, called the \emph{foot node}, is labeled with the same non-terminal as the root},
    \item each remaining leaf is labeled with a terminal, and
    \item each inner node can be marked for adjunction.
\end{itemize}
Initial trees are the starting points for derivations of the tree-adjoining grammar. These trees are then extended by repeatedly replacing nodes marked for adjunction by auxiliary trees. Formally, given an initial or auxiliary tree $t$ that contains at least one inner node $v$
marked for adjunction and given an auxiliary tree $a$ whose root $r$ has the same label as
$v$, we can combine these trees with the following operation called {\em adjunction}, see Figure~\ref{fg:adj} for an example.
\begin{enumerate}
    \item Replace $a$'s foot node by the subtree rooted at $v$.
    \item Replace the node $v$ with the tree obtained from the last step.
\end{enumerate}
Note that these steps make sense, since $r$ and $v$ have the same label.
Note that adjunction does not change the number leaves labeled with a non-terminal symbol,
i.e., an initial tree will stay an initial tree and an auxiliary tree will stay an auxiliary tree.

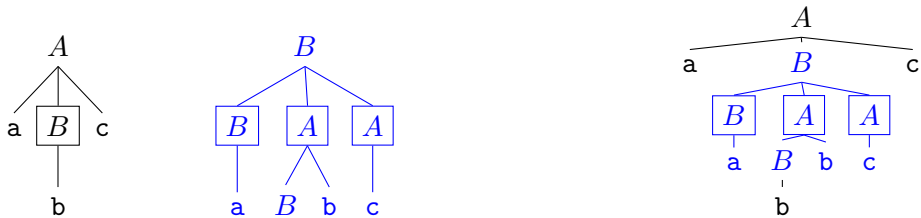
\begin{figure}[ht]
    \begin{subfigure}[b]{.47\textwidth}
        \begin{center}
            \begin{tikzpicture}
                \Tree [
                .$A$ {\tt a} [
                .\node[draw]{$B$}; {\tt b}
                ] {\tt c}
                ]
            \end{tikzpicture}
            \quad\quad\quad
            \begin{tikzpicture}\color{blue}
                \Tree [
                .$B$ [ .\node[draw]{$B$}; [ .{\tt a} ] ] [ .\node[draw]{$A$}; [ .$B$ ] {\tt b} ]
                [ .\node[draw]{$A$}; {\tt c} ]
                ]
            \end{tikzpicture}
        \end{center}
        \caption{An initial tree (left) and an auxiliary tree (right); the internal nodes labeled $A$ and
        $B$ are marked for adjunction.}
        \label{fg:initialauxiliary}
    \end{subfigure}
    ~$\quad$~
    \begin{subfigure}[b]{.47\textwidth}
        \begin{center}
            \begin{tikzpicture}\color{blue}
                \tikzset{level 1/.style={level distance=17pt}}
                \tikzset{level 2/.style={level distance=20pt}}
                \tikzset{level 3/.style={level distance=17pt}}
                \tikzset{level 4/.style={level distance=17pt}}
                \Tree [
                .\color{black}$A$ \edge[color=black]; {\color{black}\tt a} \edge[color=black]; [
                .$B$ [ .\node[draw]{$B$}; [ .{\tt a} ] ]
                [ .\node[draw]{$A$}; [ .$B$ \edge[color=black]; [
                .\color{black}{\tt b} ]] {\tt b} ]
                [ .\node[draw]{$A$}; {\tt c} ]
                ] \edge[color=black]; {\color{black}\tt c}
                ]
            \end{tikzpicture}
        \end{center}
        \caption{Resulting tree after adjoining the auxiliary tree into the initial tree.\\}
    \end{subfigure}
    \caption{The basic building blocks and operation of tree-adjoining grammars.}\label{fg:adj}
\end{figure}

A \emph{tree-adjoining grammar} is now defined as a tuple $\Gamma = (I, A, T, N)$ where
\begin{itemize}
    \item $I$ is a finite set of initial trees and
    \item $A$ is a finite set of auxiliary trees,
\end{itemize}
using the same terminals $T$ and non-terminals $N$ as labels. The set $D$ of \emph{derived trees} of $\Gamma$ consists of all trees that can be generated by starting with an initial tree in $I$ and repeatedly adjoining auxiliary trees in $A$. (Note that each derived tree is also an initial tree, but not necessarily in $I$.)
Finally, a string $s$ over alphabet $T$ can be \emph{generated} by $\Gamma$, if there is a derived tree $t$ in $D$ such that
\begin{itemize}
    \item $t$ contains no nodes marked for adjunction and
    \item $s$ is obtained by concatenating the labels of the leaves of $t$ from left to
        right.
\end{itemize}
The language $L(\Gamma)$ is then the set of all strings that can be generated by $\Gamma$.

\section{Encoding a graph in a string} \label{sec:string}

Given a graph $G=(V,E)$, in this section we construct a string $\GG{G}$ (the \emph{graph gadget}) that encodes its $k$-cliques, over the terminal alphabet $T = \{0,1,\$,\#,|,\S,e,l_1,\ldots,l_6,r_1,\ldots,r_6\}$ of size 19.
In the next section we then design a tree-adjoining grammar $\Gamma$ that generates $\GG{G}$ if and only if $G$ contains a $6k$-clique.
We assume that $V = \{1,\ldots,|V|\}$, and we denote the binary representation of any $v \in V$ by $\overline{v}$ and the neighborhood of $v$ by $N(v)$. For two strings $a$ and $b$, we use $a\circ b$ to denote their concatenation and $a^R$ to denote the reverse of $a$.

We start with \emph{node} and \emph{list gadgets}, encoding a vertex and its neighborhood, respectively:
\[
    \NG{v} := \$\,\overline{v}\,\$\quad\text{and}\quad
    \LG{v} := \underset{u \in  N(v)}{\bigcirc}\NG{u}\;= \underset{u \in  N(v)}{\bigcirc} \$\, \overline{u} \,\$
\]
Note that $u$ and $v$ are adjacent iff $\NG{u}$ is a substring of $\LG{v}$.

Next, we build clique versions of these gadgets, that encode a \emph{$k$-clique} $C$ and its neighborhood, respectively:
\[
    \CNG{C} := \underset{v \in  C}{\bigcirc} {\left( \#\;\NG{v}\;
    \#\right)}^k\quad\text{and}\quad
    \CLG{C} := {\left(\underset{v \in  C}{\bigcirc} \#\;\LG{v}\;
    \#\right)}^k
\]
Note that two $k$-cliques $C$ and $C'$ form a $2k$-clique
if and only if the substring of $\CNG{C}$ between the $i$-th and $(i+1)$-th symbol $\#$ is a substring of 
the substring of $\CLG{C'}$ between the $i$-th and $(i+1)$-th symbol $\#$, for all $i$. Indeed, every pair of a vertex in $C$ and a vertex in $C'$ is tested for adjacency. 

Conceptually, we split any $6k$-clique into six $k$-cliques. Thus, let $\mathcal{C}_k$ be the set of all $k$-cliques in $G$. Our final encoding of the graph is:
\begin{alignat*}{4}
    \GG{G} :=& \underset{C \in  \mathcal{C}_k}{\bigcirc}
             &&|\;\blue{\CNG{C}\;\S\;\CLG{C}^R}\;
             &&l_1\;r_1\;\red{\CLG{C}\;\S\; \CLG{C}^R}\;
             &&| \\[-.5em]
    \circ& \underset{C \in  \mathcal{C}_k}{\bigcirc}
         &&|\; \red{\CNG{C} \; \S \; \CLG{C}^R}\;
         &&l_2\;r_2\;\purple{\CLG{C}\; \S \; \CLG{C}^R}\;
         &&| \\[-.5em]
    \circ&  \underset{C \in  \mathcal{C}_k}{\bigcirc}
         &&|\; \purple{\CNG{C} \; \S\; \CLG{C}^R}\;
         &&l_3\;r_3\; \blue{\CLG{C} \; \S\; \CLG{C}^R }\;
         &&| \\[-.5em]
    \circ& \; e\\[-.5em]
    \circ&  \underset{C \in  \mathcal{C}_k}{\bigcirc}
         &&|\; \blue{\CLG{C} \; \S \; \CLG{C}^R }\;
         &&l_4\;r_4\; \purple{\CNG{C} \; \S \; \CLG{C}^R }\;
         && |\\[-.5em]
    \circ& \underset{C \in  \mathcal{C}_k}{\bigcirc}
         &&|\; \purple{ \CLG{C} \; \S \; \CLG{C}^R }\;
         &&l_5\;r_5\; \red{\CNG{C} \; \S \; \CLG{C}^R}\;
         &&| \\[-.5em]
    \circ& \underset{C \in  \mathcal{C}_k}{\bigcirc}
         &&|\; \red{\CLG{C}\; \S \; \CLG{C}^R } \;
         &&l_6\;r_6\; \blue{\CNG{C}\; \S \; \CLG{C}^R }\;
         &&|\end{alignat*}
    As we will show, there is a tree-adjoining grammar of constant size that generates
    the string $\GG[k]{G}$ iff $G$ contains a $6k$-clique. The structure of this test is depicted in Figure~\ref{fg:clique}.
    The clique-gadgets of the same
    highlighting style together allow us to test for an almost-$4k$-clique, as it is depicted in Figure~\ref{fg:decomposealmostfourclique}. The two gadgets of the same highlighting style then test for two claws of cliques, as depicted in Figure~\ref{fg:claw}.
    \begin{figure}[ht]
        \begin{subfigure}[b]{.47\textwidth}
            \begin{center}
                \begin{tikzpicture}
                    \node[draw, circle] (1) at (0,0) {$C_1$};
                    \node[draw, circle] (6) at (0,-1.2) {$C_6$};
                    \node[draw, circle] (2) at (1.2,1) {$C_2$};
                    \node[draw, circle] (5) at (1.2,-2.1) {$C_5$};
                    \node[draw, circle] (3) at (2.4,0) {$C_3$};
                    \node[draw, circle] (4) at (2.4,-1.2) {$C_4$};

                    \draw[dashed, color=blue] (1) to (3);
                    \draw[dashed, color=blue] (1) to (4);
                    \draw[dashed, color=blue] (1) to (6);
                    \draw[dashed, color=blue] (6) to (3);
                    \draw[dashed, color=blue] (6) to (4);

                    \draw[dotted, thick, color=violet] (3) to (2);
                    \draw[dotted, thick, color=violet] (3) to (4);
                    \draw[dotted, thick, color=violet] (3) to (5);
                    \draw[dotted, thick, color=violet] (4) to (2);
                    \draw[dotted, thick, color=violet] (4) to (5);

                    \draw[color=red] (2) to (1);
                    \draw[color=red] (2) to (5);
                    \draw[color=red] (2) to (6);
                    \draw[color=red] (5) to (1);
                    \draw[color=red] (5) to (6);
                \end{tikzpicture}
            \end{center}
            \caption{Each $C_i$ is a $k$-clique and there is an edge between two $k$-cliques of some
                highlighting style if the clique gadgets of that style ensure that these two cliques
            together form a $2k$-clique.}
            \label{fg:decomposealmostfourclique}
        \end{subfigure}
    ~$\quad$~
        \begin{subfigure}[b]{.47\textwidth}
            \begin{center}
                \begin{tikzpicture}
                    \node[draw, circle] (1) at (0,0) {$C_1$};
                    \node[draw, circle] (6) at (0,-1.2) {$C_6$};
                    \node[draw, circle] (2) at (1.2,1) {$C_2$};
                    \node[draw, circle] (5) at (1.2,-2.1) {$C_5$};

                    \draw (2) to (1);
                    \draw (2) to (5);
                    \draw (2) to (6);
                    \draw[dashed] (5) to (1);
                    \draw[dashed] (5) to (6);
                    \draw[dashed, bend right] (5) to (2);
                \end{tikzpicture}
            \end{center}
            \caption{We will generate an almost-$4k$-clique as in (a) by generating two claws.
            (This tests the edges $(C_1, C_6)$, $(C_2, C_5)$, and $(C_3, C_4)$ in (a) twice.)\\}\label{fg:claw}
        \end{subfigure}
        \caption{Structure of our test for $6k$-cliques.}\label{fg:clique}
    \end{figure}

    As the  graph has $n$ nodes, for any node $u$ the node and list gadgets $\NG{u}, \LG{u}$ have a length of $O(n \log n)$, and for a $k$-clique $C$ the clique neighborhood gadgets $\CNG{C}, \CLG{C}$ thus have a length of
    $O(k^2 n \log n)$. As our encoding of the graph consists
    of $O(n^k)$ clique neighborhood gadgets, the resulting string length is
    $O(k^2 n^{k+1} \log n)=O(n^{k+1}\log n)$.
    It is easy to see that it is also possible to construct all gadgets and in particular the
    encoding of a graph in linear time with respect to their length.

    \section{Programming with trees} \label{sec:programming}

    It remains to design a clique-detecting tree-adjoining grammar. To make our reduction more accessible,
    we will think of tree-adjoining grammars as a certain programming language. In the end,
    we will then present a “program” that generates (a suitable superset of) the set all strings
    that represent a graph containing a $6k$-clique.
    We start by defining programs.

    A {\em normal tree} $N$ with {\em input} $N_{In}$ and {\em output} $N_{Out}$ is an auxiliary tree where:
    \begin{itemize}
        \item the root is labeled with $N_{In}$,
        \item exactly one node is marked for adjunction, and
        \item this node lies on the path from the root to the foot node and is labeled $N_{Out}$.
    \end{itemize}
    See Figure~\ref{fg:normaltree} for an illustration.
    The special structure of a normal tree $N$ allows us to split its nodes into four categories (excluding the path from $N$'s root to its foot node): subtrees of left children of the path from $N$'s root to $N_{Out}$, subtrees of left children of the path from $N_{Out}$ to $N$'s foot node, subtrees of right children of the path from $N_{Out}$ to $N$'s foot node, and the remaining nodes (i.e., subtrees of right children of the path from $N$'s root to $N_{Out}$).
    The concatenation of all terminal symbols in $N$'s leaves from left to right can then be split into four
    parts $n_1, n_2, n_3, n_4$ where each part contains symbols from exactly one category.
    We say that the normal tree $N$ \emph{generates} the tuple $(n_1,n_2,n_3,n_4)$.
    \begin{lemma}\label{lm:4:ntc}
        Given normal trees $N$ with input $N_{In}$, output $N_{Out}$ and $M$ with input $M_{In} = N_{Out}$, output $M_{Out}$, the derived tree $N\cdot M$ obtained by adjoining $M$ into $N$ is
        a normal tree with input $N_{In}$ and output $M_{Out}$. Further, if $N$ and $M$ generate the
        tuples $(n_1,n_2,n_3,n_4)$ and $(m_1,m_2,m_3,m_4)$, then $N\cdot M$ generates the tuple
        $(n_1\circ m_1, m_2\circ n_2, n_3\circ m_3, m_4\circ n_4)$.
    \end{lemma}
    \begin{proof}
        See Figure~\ref{fg:ntc}.
    \end{proof}
    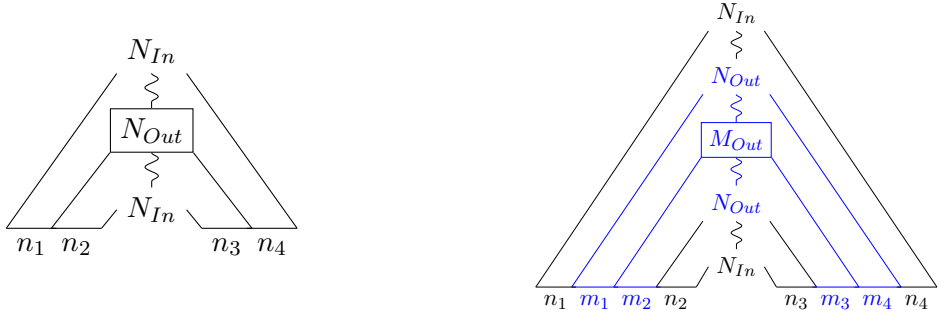
\begin{figure}[ht]
        \begin{subfigure}[b]{0.49\textwidth}
            \begin{center}
                \begin{tikzpicture}
                    \Tree[
                    .\node (1) {$N_{In}$}; \edge[decorate, decoration={snake, segment length=6.7pt}]; [
                    .\node[draw] (2) {$N_{Out}$}; \edge[decorate, decoration={snake, segment length=6.7pt}];[
                    .\node (3) {$N_{In}$};
                    ]
                    ]
                    ]
                    \node (4) at (-1,-2.5) {$n_2$};
                    \draw (4.north east) to (4.north west);
                    \draw (4.north east) to (3.west);
                    \draw (2.south west) to (4.north west);

                    \node (5) at (-1.6,-2.5) {$n_1$};
                    \draw (5.north east) to (5.north west);
                    \draw (1.south west) to (5.north west);

                    \node (6) at (1,-2.5) {$n_3$};
                    \draw (6.north west) to (6.north east);
                    \draw (6.north west) to (3.east);
                    \draw (2.south east) to (6.north east);

                    \node (7) at (1.6,-2.5) {$n_4$};
                    \draw (7.north west) to (7.north east);
                    \draw (1.south east) to (7.north east);
                    \node (A) at (1,-3.3) {};
                \end{tikzpicture}
            \end{center}
            \caption{A normal tree $N$.}
            \label{fg:normaltree}
        \end{subfigure}
        ~
        \begin{subfigure}[b]{0.49\textwidth}
            \begin{center}
                \scalebox{0.8}{
                    \begin{tikzpicture}
                        \Tree [
                        .\node (0) {$N_{In}$}; \edge[decorate, decoration={snake, segment length=6.7pt}]; [
                        .\node[color=blue] (1) {$N_{Out}$}; \edge[color=blue,decorate, decoration={snake, segment length=6.7pt}]; [
                        .\node[draw, color=blue] (2) {$M_{Out}$}; \edge[color=blue,decorate, decoration={snake, segment length=6.7pt}];[
                        .\node[color=blue] (3) {$N_{Out}$};  \edge[decorate, decoration={snake, segment length=6.7pt}]; [
                        .\node (35) {$N_{In}$};
                        ]
                        ]
                        ]
                        ]
                        ]
                        \node (8) at (-1,-4.7) {$n_2$};
                        \draw (8.north east) to (8.north west);
                        \draw (8.north east) to (35.west);
                        \draw (3.south west) to (8.north west);

                        \node[color=blue] (4) at (-1.65,-4.7) {$m_2$};
                        \draw[color=blue] (4.north east) to (4.north west);
                        \draw[color=blue] (2.south west) to (4.north west);

                        \node[color=blue] (5) at (-2.35,-4.7) {$m_1$};
                        \draw[color=blue] (5.north east) to (5.north west);
                        \draw[color=blue] (1.south west) to (5.north west);

                        \node (9) at (-3,-4.7) {$n_1$};
                        \draw (9.north east) to (9.north west);
                        \draw (0.south west) to (9.north west);

                        \node (10) at (1,-4.7) {$n_3$};
                        \draw (10.north east) to (10.north west);
                        \draw (10.north west) to (35.east);
                        \draw (3.south east) to (10.north east);

                        \node[color=blue] (11) at (1.65,-4.7) {$m_3$};
                        \draw[color=blue] (11.north east) to (11.north west);
                        \draw[color=blue] (2.south east) to (11.north east);

                        \node[color=blue] (12) at (2.35,-4.7) {$m_4$};
                        \draw[color=blue] (12.north east) to (12.north west);
                        \draw[color=blue] (1.south east) to (12.north east);

                        \node (13) at (3,-4.7) {$n_4$};
                        \draw (13.north east) to (13.north west);
                        \draw (0.south east) to (13.north east);
                    \end{tikzpicture}
                }
            \end{center}
            \caption{The tree resulting after adjoining $M$ into $N$.}
        \end{subfigure}
        \caption{Adjoining normal trees.}\label{fg:ntc}
    \end{figure}
    We now define a {\em program ${\sf P}$ with input ${\sf P}_{In}$ and output ${\sf P}_{Out}$}  as a set of normal trees
    that contains a tree with input ${\sf P}_{In}$ and a tree with output ${\sf P}_{Out}$.
    Note that all trees derived by starting with a tree in ${\sf P}$ and repeatedly adjoining trees from ${\sf P}$ are normal, by Lemma~\ref{lm:4:ntc}.
    An {\em execution} of the program ${\sf P}$ is a derived
    tree of {\sf P} with input ${\sf P}_{In}$ and output ${\sf P}_{Out}$. Further, the {\em set
    computed by {\sffamily P}}, denoted by $L({\sf P})$, is the set of all tuples generated by {\sf P}'s executions.

    We will later use programs as subroutines of tree-adjoining grammars. Let $N({\sf P})$ be the set of non-terminals of {\sf P}. Formally, we say that ${\sf P}$ is a \emph{subroutine} of a grammar $\Gamma$ if
    \begin{itemize}
        \item the set of trees ${\sf P}$ is a subset of the auxiliary trees of $\Gamma$, and
        \item no remaining auxiliary tree of $\Gamma$ has a root label in $N({\sf P}) \setminus \{{\sf P}_{Out}\}$.
    \end{itemize}
    These restrictions ensure that any ``call'' to the program {\sf P} terminates at ${\sf P}_{Out}$.
    Indeed, consider any sequence of adjunctions in $\Gamma$ ending in a tree without nodes marked for adjunction. If this sequence contains an adjunction of a node labeled ${\sf P}_{In}$, meaning that program {\sf P} is called, then this adjunction must be followed by an execution of ${\sf P}$, i.e., it must generate a derived tree of {\sf P} with output ${\sf P}_{Out}$. Indeed, any derived tree of {\sf P} is normal and thus contains exactly one node marked for adjunction. To get rid of this node, we have to adjoin some auxiliary tree, but the remaining auxiliary trees can only adjoin to ${\sf P}_{Out}$.
    We will frequently make use of this observation that ensures coherence of programs.

    We now show how to perform two programs sequentially one after another. To avoid interference, we ensure that the two programs have disjoint non-terminals, except for their input and output. In particular, we will model two sequential calls to the same program by creating two copies of the program.

    \begin{lemma}[Combining programs]\label{lm:4:combine}
        For programs ${\sf P}$ and ${\sf Q}$, let ${\sf Q'}$
        denote the program obtained from {\sf Q} by replacing each non-terminal by a fresh copy, ensuring that {\sf P} and ${\sf Q'}$ have disjoint non-terminals. Further, let ${\sf Q''}$ denote the program obtained from ${\sf Q'}$ by replacing ${\sf Q'}_{In}$ by ${\sf P}_{Out}$.
        Then ${\sf P}\cdot{\sf Q} := {\sf P}\cup {\sf Q''}$ is a program computing the set
        \[ L({\sf P}\cdot{\sf Q}):=\{(a\circ a',b'\circ b, c\circ c', d'\circ d)\mid (a,b,c,d)\in L({\sf P}), (a',b',c',d')\in L({\sf Q})\}. \]
    \end{lemma}
    \begin{proof}
        As every execution of {\sf P} and {\sf Q''} is a normal tree, the claim follows
        from Lemma~\ref{lm:4:ntc}.
    \end{proof}
    We can think of $\cdot$ as an operator on programs; the above lemma shows that it is associative.

    \subsection{Basic programs}

    We now present some easy programs that will later be used as subroutines.

    \paragraph*{Writing characters} We start by demonstrating a program that writes exactly one
    character to each of the four positions.
    Formally, given a 4-tuple of characters $(a,b,c,d)$, let the program ${\sf W}{(a,b,c,d)}$ be
    defined by the following auxiliary tree:
    \begin{center}
        \scalebox{1}
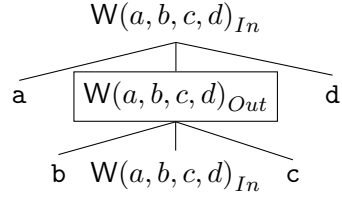
{
            \begin{tikzpicture}
                \Tree [
                .${\sf W}{(a,b,c,d)}_{In}$ {\tt a} [
                .\node[draw]{${\sf W}{(a,b,c,d)}_{Out}$}; {\tt b} ${\sf W}{(a,b,c,d)}_{In}$ {\tt c}
                ] {\tt d}
                ]
        \end{tikzpicture}}
    \end{center}
    Clearly, this tree is normal with input ${\sf W}{(a,b,c,d)}_{In}$ and output ${\sf W}{(a,b,c,d)}_{Out}$, so that ${\sf W}{(a,b,c,d)}$ is a program. The tree itself is an execution of the program, and it is the only execution. Thus, this program computes the set $L({\sf W}{(a,b,c,d)}) = \{(a,b,c,d)\}$. We write ${\sf W}(a)$ to denote the program ${\sf W}{(a,a,a,a)}$.

    \paragraph*{Testing equality} We give a program that tests equality of four strings, by writing the same arbitrary string to all four positions.
    Formally, for any terminal alphabet $\Sigma$, let the program ${\sf Eq}(\Sigma)$ be defined by the
    following set of $|\Sigma| + 1$ auxiliary trees:
    \begin{center}
        \scalebox{1}
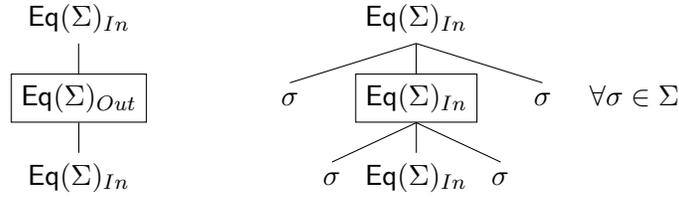
{
            \begin{tikzpicture}
                \Tree [
                .${\sf Eq}(\Sigma)_{In}$ [
                .\node[draw]{${\sf Eq}(\Sigma)_{Out}$}; ${\sf Eq}(\Sigma)_{In}$
                ]
                ]
            \end{tikzpicture}
            \quad\quad\quad\quad
            \begin{tikzpicture}
                \Tree [
                .${\sf Eq}(\Sigma)_{In}$ $\sigma$  [
                .\node[draw]{${\sf Eq}(\Sigma)_{In}$}; $\sigma$  ${\sf Eq}(\Sigma)_{In}$ $\sigma$
                ] $\sigma$
                ]
                \node () at (2.9,-.93) {$\forall\sigma \in \Sigma$};
        \end{tikzpicture}}
    \end{center}
    A simple induction shows that $L({\sf Eq}(\Sigma)) = \{ (v, v^R, v, v^R) \mid v \in  \Sigma^* \}$.

    \paragraph*{Writing anything}
    We will need to write appropriate strings surrounding some carefully
    constructed substrings. As it turns out, being able to write anything will be sufficient; this
    is achieved by the following program. Given an alphabet $\Sigma$,
    let the program ${\sf A}(\Sigma)$ be defined by the following set of $4|\Sigma|+1$ trees:
    \begin{center}
        \scalebox{1}
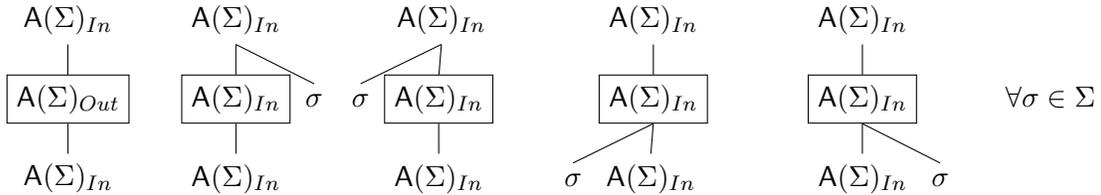
{
            \begin{tikzpicture}
                \Tree [
                .${\sf A}(\Sigma)_{In}$ [
                .\node[draw]{${\sf A}(\Sigma)_{Out}$}; ${\sf A}(\Sigma)_{In}$
                ]
                ]
            \end{tikzpicture}
            \begin{tikzpicture}
                \Tree [
                .${\sf A}(\Sigma)_{In}$ \edge[draw=none]; \phantom{$\sigma$} [
                .\node[draw]{${\sf A}(\Sigma)_{In}$}; ${\sf A}(\Sigma)_{In}$
                ] $\sigma$
                ]
            \end{tikzpicture}
            \begin{tikzpicture}
                \Tree [
                .${\sf A}(\Sigma)_{In}$ $\sigma$  [
                .\node[draw]{${\sf A}(\Sigma)_{In}$}; ${\sf A}(\Sigma)_{In}$
                ] \edge[draw=none]; \phantom{$\sigma$ }
                ]
            \end{tikzpicture}
            \begin{tikzpicture}
                \Tree [
                .${\sf A}(\Sigma)_{In}$  [
                .\node[draw]{${\sf A}(\Sigma)_{In}$}; $\sigma$ ${\sf A}(\Sigma)_{In}$ \edge[draw=none]; \phantom{$\sigma$ }
                ]
                ]
            \end{tikzpicture}
            \begin{tikzpicture}
                \Tree [
                .${\sf A}(\Sigma)_{In}$ [
                .\node[draw]{${\sf A}(\Sigma)_{In}$}; \edge[draw=none]; \phantom{$\sigma$} ${\sf A}(\Sigma)_{In}$ $\sigma$                ]
                ]
                \node () at (2.5,-.93) {$\forall\sigma \in \Sigma$};
        \end{tikzpicture}}
    \end{center}
    As this program allows writing anything, it is easy to see that
    ${\sf A}(\Sigma)$ computes the set $L({\sf A}(\Sigma)) = \{(v_1,v_2,v_3,v_4) \mid v_1,v_2,v_3,v_4 \in \Sigma^* \} = {\left(\Sigma^*\right)}^4$.

    \subsection{Detecting Cliques}

    With the help of the above programs, we now design programs that detect a $6k$-clique.

    \paragraph*{Detecting claws}
    Our next program can detect whether four nodes form a claw graph.
    \[
        {\sf NC} := {\sf W}(\#) \cdot {\sf A}(\{0,1,\$\}) \cdot {\sf W}(\$) \cdot {\sf Eq}(\{0,1\})
        \cdot {\sf W}(\$) \cdot {\sf A}(\{0,1,\$\}) \cdot {\sf W}{(\#)}
    \]
    \begin{lemma}\label{lm:4:nc}
        For any nodes $v_1,v_2,v_3,v_4$, the program {\sf NC} generates the tuple
        \[(a,b,c,d) := (\# \, \NG{v_1} \, \# ,  \, \# \, \LG{v_2}^R \, \# , \,  \# \, \LG{v_3} \, \# ,  \, \# \, \LG{v_4}^R \, \# )\]
        and any of its cyclic rotations (i.e., $(b,c,d,a)$, $(c,d,a,b)$, and $(d,a,b,c)$)
        if and only if $v_1$ is adjacent to each one of $v_2, v_3$, and $v_4$.
    \end{lemma}
    \begin{proof}
        By Lemma~\ref{lm:4:combine} and the properties of basic programs,
        we see that {\sf NC} computes all tuples of the form\[
            ( \#\;\alpha_1\;\$\;v\;\$\;\alpha_2\;\#, \#\;\alpha_3\;\$\;v^R\;\$\;\alpha_4\;\#,\#
            \;\alpha_5\;\$\;v\;\$\;\alpha_6\;\#,\#\;\alpha_7\;\$\;v^R\;\$\;\alpha_8\;\#)
        \]
        where $v\in\{0,1\}^*$ and $\alpha_1,\ldots, \alpha_8\in\{0,1,\$\}^*$.
        From the construction of node and list gadgets we see that all tuples $( \# \NG{v_1} \# , \# \LG{v_2}^R \# , \# \LG{v_3} \# ,  \# \LG{v_4}^R \# )$ are of this form.

        For the other direction, for any generated tuple $(a,b,c,d)$, where $a$ is $\#\,\$\,v\,
        \$\,\#$, it holds that $\$\,v\,\$$ is a substring of $c$ and its reverse is a substring of $b$ and $d$.
        Hence, $\NG{v_1}$ is a substring of $\LG{v_2}, \LG{v_3}$, and $\LG{v_4}$. This implies that $v_1$ is adjacent to $v_2, v_3$, and $v_4$.
    \end{proof}

    \paragraph*{Detecting claws of cliques}
    We now extend {\sf NC} to a program that can detect claws of $k$-cliques, see Figure~\ref{fg:claw}.
    We define the program ${\sf CC}$ by the following
    set of 3 trees (additional to the trees of $\sf NC$):
    \begin{center}
        \scalebox{1}
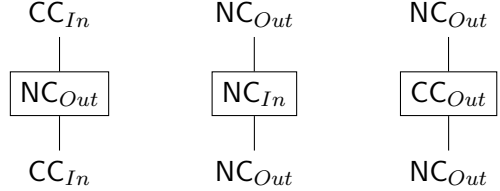
{
            \begin{tikzpicture}
                \Tree [
                .${\sf CC}_{In}$ [
                .\node[draw] {${\sf NC}_{Out}$}; [
                .${\sf CC}_{In}$
                ]
                ]
                ]
            \end{tikzpicture}
            \quad\quad
            \quad
            \begin{tikzpicture}
                \Tree [
                .${\sf NC}_{Out}$ [
                .\node[draw]
                {${\sf NC}_{In}$}; [
                .${\sf NC}_{Out}$
                ]
                ]
                ]
            \end{tikzpicture}
            \quad\quad
            \quad
            \begin{tikzpicture}
                \Tree [
                .${\sf NC}_{Out}$ [
                .\node[draw] {${\sf CC}_{Out}$}; [
                .${\sf NC}_{Out}$
                ]
                ]
                ]
        \end{tikzpicture}}
    \end{center}
    Each execution of {\sf CC} starts with the first tree, then repeatedly adjoins the second tree followed by some execution of {\sf NC}, and finally adjoins the last tree. As the number of repetitions is arbitrary, the program {\sf CC} can perform any number of sequential calls to {\sf NC}.\footnote{Actually, we already know how many calls to {\sf NC} we want to perform, namely $k^2$. However, encoding this number into the grammar would result in a grammar size depending on $k$, which we want to avoid.}
    \begin{lemma}\label{lm:4:fork}
        For any $k$-cliques $C_1, C_2, C_3, C_4$ in $G$, the program {\sf CC} generates the tuple
        $(a,b,c,d) := (\CNG{C_1}, \CLG{C_2}^R, \CLG{C_3}, \CLG{C_4}^R)$ and all of its cyclic rotations (i.e., $(b,c,d,a)$, $(c,d,a,b)$, and $(d,a,b,c)$)
        if and only if $C_1 \cup C_2$, $C_1 \cup C_3$, and $C_1 \cup C_4$ each form a $2k$-clique in~$G$.
    \end{lemma}
    \begin{proof}
        For any nodes $v_i^j$, with $i \in [4], j \in [m], m \ge 1$, set
        \[ n_i := \underset{j \in [m]}{\bigcirc} \# \, NG( v_i^j ) \, \# \qquad \text{and} \qquad \ell_i := \underset{j \in [m]}{\bigcirc} \# \, LG( v_i^j ) \, \#. \]
        As program {\sf CC} can perform any number of calls to {\sf NC}, and by Lemma~\ref{lm:4:nc}, program {\sf CC} generates the tuple $(n_1,\ell_2,\ell_3,\ell_4)$ if and only if $v_1^j$ is adjacent to $v_2^j, v_3^j$, and $v_4^j$ for all $j$.

        Observe that for any $k$-cliques $C_1 = \{v_1,\ldots,v_k\}, C_2 = \{u_1,\ldots,u_k\}$, both
        $\CNG{C}$ and $\CLG{C}$ can be split into $k^2$ blocks by splitting
        between two consecutive $\#$-characters:
        \begin{alignat*}{20}
            \CNG{C_1}  =   &\;& \#\,&&&\NG{v_1}&&\,\#&
            \#\,&&&\NG{v_1}&&\,\#\cdots\,&
            \#\,&&&\NG{v_1}&&\,\#&
            \#\,&&&\NG{v_2}&&\,\#\cdots\nonumber\\
            \CLG{C_2} = &&      \#\,&&&\LG{u_1} &&\,\#&
            \#\,&&&\LG{u_2} &&\,\#\cdots&
            \#\,&&&\LG{u_k} &&\,\#&
            \#\,&&&\LG{u_1} &&\,\#\cdots
        \end{alignat*}
        This layout is chosen so that each node $v_i$ in $C_1$ is paired up with each node $u_j$ in $C_2$ exactly once. The claim follows from these two insights.
    \end{proof}

    \paragraph*{Detecting almost-\boldmath$4k$-cliques}
    We now use {\sf CC} twice to test for two claws, thus detecting ``almost-$4k$-cliques'', as depicted in Figure~\ref{fg:claw}:
    \[ {\sf C} := {\sf CC}
        \cdot {\sf W}(\S)
    \cdot {\sf CC}. \]
    Lemmas~\ref{lm:4:fork} and \ref{lm:4:combine} directly imply the following, see Figure~\ref{fg:claw}.
    \begin{lemma}\label{lm:4:almostclique}
        For any $k$-cliques $C_a,C_b,C_c,C_d$ the program ${\sf C}$ generates
        the tuple
        {\small \[
                (\CNG{C_a}\; \S \; \CLG{C_a}^R,
                \CLG{C_b}\;  \S \; \CLG{C_b}^R,
                \CLG{C_c}\;  \S \; \CLG{C_c}^R,
                \CNG{C_d}\;  \S \; \CLG{C_d}^R  )
        \]}%
        if and only if $C_a\cup C_b\cup C_d$ and
        $C_a\cup C_c\cup C_d$ both form a $3k$-clique.
        A similar statement holds if we pick any two other positions in the tuple for the $\CNG{\cdot}$ gadgets.
    \end{lemma}

    \paragraph*{Detecting \boldmath$6k$-cliques}
    As in Figure~\ref{fg:decomposealmostfourclique}, we now want to test for three almost-$4k$-cliques to detect a $6k$-clique.
    Recall that $T = \{0,1,\$,\#,|,\S,e,l_1,\ldots,l_6,r_1,\ldots,r_6\}$ is the terminal alphabet that we constructed our strings over.
    The following programs will generate the highlighted
    groups in Figure~\ref{fg:decomposealmostfourclique}:
    \begin{align*}
        {\sf P(1,3,4,6)} &:= {\sf A}(T)\,\cdot\,{\sf W}{(|)}\,\cdot\,{\sf C}\,\cdot\,{\sf W}(l_1,r_3,l_4,r_6)\\
        {\sf P(1,2,5,6)} &:= {\sf W}{(r_1,l_2,r_5,l_6)}\,\cdot\,{\sf C}\,
        \cdot\,{\sf W}{(|)}\,\cdot\,{\sf A}(T),\\
        {\sf P(2,3,4,5)} &:= {\sf W}(r_2,l_3,r_4,l_5)\,\cdot\,{\sf C}\,
        \cdot\,{\sf W}{(|)}\,\cdot\,{\sf A}(T),
    \end{align*}
    We now deviate from our notion of normal trees by explicitly \emph{not} marking ${\sf P(1,2,5,6)}_{Out}$ and ${\sf P(2,3,4,5)}_{Out}$ for adjunction.
    Our final tree-adjoining grammar $\Gamma$ consists of the following initial and auxiliary trees (as well as all auxiliary trees used by its subroutines):
    \begin{center}
        \scalebox{1}
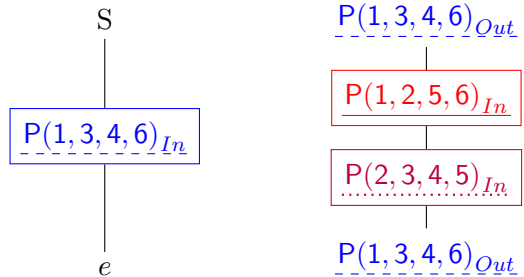
{
            \begin{tikzpicture}
                \tikzset{level 1/.style={level distance=43pt}}
                \tikzset{level 2/.style={level distance=50pt}}
                \Tree [.S [.\node[draw, color=blue]{\udash{${\sf P(1,3,4,6)}_{In}$}}; [
                .$e$  ] ] ]
            \end{tikzpicture}
            \quad\quad\quad\quad
            \begin{tikzpicture}
                \Tree [
                .\node[color=blue]{\udash{${\sf P(1,3,4,6)}_{Out}$}}; [
                .\node[draw, color=red]{\uline{${\sf P(1,2,5,6)}_{In}$}}; [
                .\node[draw, color=purple]{\udot{${\sf P(2,3,4,5)}_{In}$}}; [
                .\node[color=blue]{\udash{${\sf P(1,3,4,6)}_{Out}$}};
                ]
                ]
                ]
                ]
        \end{tikzpicture}}
    \end{center}
    Note that the latter tree is the only one in $\Gamma$ that has more than one node marked for adjunction, so it needs special treatment in the analysis.
    \begin{lemma}
        For any graph $G$, the grammar $\Gamma$ generates the encoding $\GG{G}$ if and only if $G$ contains a $6k$-clique. Moreover, $\Gamma$ has constant size (independent of $k$).
    \end{lemma}
    \begin{figure}[ht]
        \begin{center}
            \scalebox{0.58}{
                \begin{tikzpicture}
                    \Tree [
                    .S [
                    .\node[color=blue](1){\udash{${\sf P(1,3,4,6)}_{In}$}};
                    \edge[color=blue, decorate, decoration={snake, segment length=6.1pt}]; [
                    .\node[color=blue](2){\udash{${\sf P(1,3,4,6)}_{Out}$}}; [
                    .\node[color=red](3){\uline{${\sf P(1,2,5,6)}_{In}$}};
                    \edge[color=red, decorate, decoration={snake, segment length=6.1pt}]; [
                    .\node[color=red](4){\uline{${\sf P(1,2,5,6)}_{Out}$}};
                    \edge[color=red, decorate, decoration={snake, segment length=6.1pt}]; [
                    .\node[color=red](5){\uline{${\sf P(1,2,5,6)}_{In}$}}; [
                    .\node[color=purple](6){\udot{${\sf P(2,3,4,5)}_{In}$}};
                    \edge[color=purple, decorate, decoration={snake, segment length=6.1pt}];[
                    .\node[color=purple](7){\udot{${\sf P(2,3,4,5)}_{Out}$}}; \edge[color=purple, decorate, decoration={snake, segment length=6.1pt}]; [
                    .\node[color=purple](8){\udot{${\sf P(2,3,4,5)}_{In}$}}; [
                    .\node[color=blue](9){\udash{${\sf P(1,3,4,6)}_{Out}$}}; \edge[color=blue, decorate, decoration={snake, segment length=6.1pt}];[
                    .\node[color=blue](10){\udash{${\sf P(1,3,4,6)}_{In}$}}; [
                    .\node(11) at (0,0) {$e$};
                    ]
                    ]
                    ]
                    ]
                    ]
                    ]
                    ]
                    ]
                    ]
                    ]
                    ]
                    ]
                    \node(12) at (-1.45,-11.49) {$r_3{\color{blue}\udash{$(C_3)$}}|\ldots $};
                    \draw[color=blue, dashed] (12.north east) to (10);
                    \draw[color=blue, dashed] (12.north west) to (9.south west);
                    \draw[color=blue, dashed] (12.north west) to (12.north east);

                    \node(13) at (-3.4,-11.49) {$\ldots |\,{\color{purple}\udot{$(C_3)$}}l_3$};
                    \draw[color=purple, dotted, thick] (13.north east) to (8.south west);
                    \draw[color=purple, dotted, thick] (13.north west) to (7.south west);
                    \draw[color=purple, dotted, thick] (13.north west) to (13.north east);

                    \node(14) at (-5.3,-11.49) {$r_2{\color{purple}\udot{$(C_2)$}}|\ldots $};
                    \draw[color=purple, dotted, thick] (14.north west) to (6.south west);
                    \draw[color=purple, dotted, thick] (14.north west) to (14.north east);

                    \node(15) at (-7.2,-11.49) {$\ldots |\,{\color{red}\uline{$(C_2)$}}l_2$};
                    \draw[color=red] (15.north east) to (5.south west);
                    \draw[color=red] (15.north west) to (4.south west);
                    \draw[color=red] (15.north west) to (15.north east);

                    \node(16) at (-9.1,-11.49) {$r_1{\color{red}\uline{$(C_1)$}}\;|\ldots $};
                    \draw[color=red] (16.north west) to (3.south west);
                    \draw[color=red] (16.north west) to (16.north east);

                    \node(17) at (-11.05,-11.49) {$\ldots |{\color{blue}\udash{$(C_1)$}}l_1$};
                    \draw[color=blue, dashed] (17.north east) to (2.south west);
                    \draw[color=blue, dashed] (17.north west) to (1.south west);
                    \draw[color=blue, dashed] (17.north west) to (17.north east);

                    \node(18) at (1.45,-11.49) {$\ldots |\,{\color{blue}\udash{$(C_4)$}}l_4$};
                    \draw[color=blue, dashed] (18.north west) to (10);
                    \draw[color=blue, dashed] (18.north east) to (9.south east);
                    \draw[color=blue, dashed] (18.north east) to (18.north west);

                    \node(19) at (3.4,-11.49) {$r_4{\color{purple}\udot{$(C_4)$}}\,|\ldots $};
                    \draw[color=purple, dotted, thick] (19.north west) to (8.south east);
                    \draw[color=purple, dotted, thick] (19.north east) to (7.south east);
                    \draw[color=purple, dotted, thick] (19.north east) to (19.north west);

                    \node(20) at (5.2,-11.49) {$\ldots |{\color{purple}\udot{$(C_5)$}}l_5$};
                    \draw[color=purple, dotted, thick] (20.north east) to (6.south east);
                    \draw[color=purple, dotted, thick] (20.north east) to (20.north west);

                    \node(21) at (7.1,-11.49) {$r_5{\color{red}\uline{$(C_5)$}}\,|\ldots $};
                    \draw[color=red] (21.north west) to (5.south east);
                    \draw[color=red] (21.north east) to (4.south east);
                    \draw[color=red] (21.north east) to (21.north west);

                    \node(22) at (9,-11.49) {$\ldots |\,{\color{red}\uline{$(C_6)$}}l_5$};
                    \draw[color=red] (22.north east) to (3.south east);
                    \draw[color=red] (22.north east) to (22.north west);

                    \node(23) at (10.89,-11.49) {$r_6{\color{blue}\udash{$(C_6)$}}|\ldots $};
                    \draw[color=blue, dashed] (23.north west) to (2.south east);
                    \draw[color=blue, dashed] (23.north east) to (1.south east);
                    \draw[color=blue, dashed] (23.north east) to (23.north west);
            \end{tikzpicture}}
        \end{center}
        \caption{Global structure of a parsing of $\GG{G}$ by $\Gamma$.
        (Clique gadgets are abbreviated.)}\label{fg:gamma}
    \end{figure}
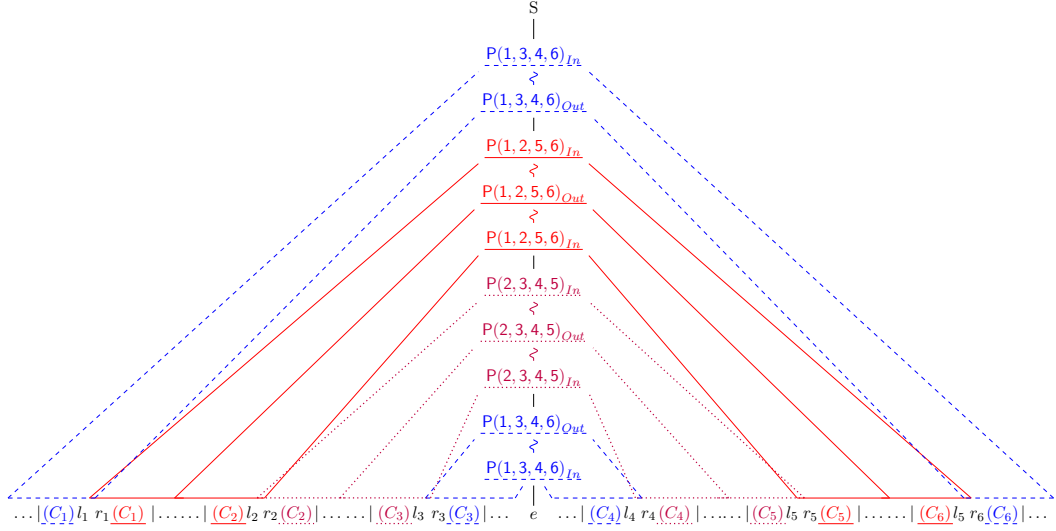
    \begin{proof}
        First, assume that $\Gamma$ can generate $\GG{G}$. Then there is a derived tree
        whose leaves, if read from left to right, yield $\GG{G}$. All derivations of $\Gamma$ start with
        the single initial tree, and then adjoin an execution of the program {\sf P(1,3,4,6)} into it.
        (As {\sf P(1,3,4,6)} is a subroutine, only a full execution can be adjoined.)
        This execution generates some tuple of strings $(\blue{x_1},\blue{x_2}, \blue{x_3}, \blue{x_4})$ and leaves exactly
        the node labeled ${\sf P(1,3,4,6)}_{Out}$ as the sole node marked for adjunction. Therefore,
        in the next step the auxiliary tree rooted with that node will be adjoined, which in turn
        leaves exactly the nodes ${\sf P(1,2,5,6)}_{In}$ and ${\sf P(2,3,4,5)}_{In}$ as nodes
        marked for adjunction. Again, these are input nodes of subroutines, therefore at both nodes
        one (complete) execution of the corresponding programs must be adjoined. The program
        execution of program ${\sf P(1,2,5,6)}$ generates a tuple of strings $(\red{y_1},\red{y_2}, \red{y_3}, \red{y_4})$, and the execution of ${\sf P(2,3,4,5)}$ generates
        $(\purple{z_1}, \purple{z_2}, \purple{z_3}, \purple{z_4})$. The grammar $\Gamma$
        ensures that these tuples will be placed in the order $(\blue{x_1},\red{y_1},\red{y_2},\purple{z_1}, \purple{z_2}, \blue{x_2},\blue{x_3}, \purple{z_3}, \purple{z_4}, \red{y_3},\red{y_4}, \blue{x_4})$, see Figure~\ref{fg:gamma} for a visualization.
        At this point, no more adjunctions are possible, since we explicitly forced ${\sf P(1,2,5,6)}_{Out}$ and ${\sf P(2,3,4,5)}_{Out}$ not to be marked for adjunction. (Also note that this structure is the only possibility to obtain a tree containing no more nodes marked for adjunction.) Hence, $\GG{G}$ can be partitioned as:\[
            \GG{G} = \blue{x_1}\circ\red{y_1}\circ\red{y_2}\circ\purple{z_1}\circ\purple{z_2}\circ\blue{x_2}\circ\blue{x_3}\circ \purple{z_3}\circ\purple{z_4}\circ\red{y_3}\circ\red{y_4}\circ\blue{x_4}.
        \]
        Consider the strings $\blue{x_1}$ and $\red{y_1}$. By the definitions of {\sf P(1,3,4,6)}
        and {\sf P(1,2,5,6)}, and Lemma~\ref{lm:4:combine}, we know that $\blue{x_1}$ must end
        with the terminal symbol $l_1$ and that $\red{y_1}$ must start with the symbol $r_1$.
        Whenever $l_1\,r_1$ occurs in $\GG{G}$, it does so in the string
        \[
            |\;\blue{\CNG{C_1}\;\S\;\CLG{C_1}^R}\;
            l_1\;r_1\;\red{\CLG{C_1}\;\S\; \CLG{C_1}^R}\;
            |,
        \]
        for some $k$-clique $C_1$.
        Since $\blue{x_1}\circ\red{y_1}$ is a substring of $\GG{G}$, and the program {\sf C} cannot produce a
        $|$-terminal, but the ${\sf W}(|)$ part of ${\sf P(\cdot,\cdot,\cdot,\cdot)}$ will always write such a $|$-character,
        $\blue{x_1}$ must have $|\;\blue{\CNG{C_1}\;\S\;\CLG{C_1}^R}\;l_1$ as a suffix and
        $\red{y_1}$ must have $r_1\;\red{\CLG{C_1}\;\S\; \CLG{C_1}^R}\;|$ as a prefix.
        This also means that the program ${\sf C}$ must generate the string between $|$ and $l_1$
        in $\blue{x_1}$ and between $r_1$ and $|$ in $\red{y_1}$.

        Similar statements hold for the other ten strings. In total we obtain that the program {\sf C}
        generates the following tuples for some $k$-cliques $C_1,\ldots,C_6$:
        \begin{itemize}
            \item $t_1 := (\blue{\CNG{C_1}\;\S\;\CLG{C_1}^R},\blue{\CLG{C_3}\;\S\;\CLG{C_3}^R},$\\
                \mbox{}$\qquad \quad\! \blue{\CLG{C_4}\;\S\;\CLG{C_4}^R}, \blue{\CNG{C_6}\;\S\;\CLG{C_6}^R})$
                in {\sf P(1,3,4,6)},
            \item $t_2 := (\red{\CLG{C_1}\;\S\;\CLG{C_1}^R}, \red{\CNG{C_2}\;\S\;\CLG{C_2}^R},$\\
                \mbox{}$\qquad \quad\! \red{\CNG{C_5}\;\S\;\CLG{C_5}^R}, \red{\CLG{C_6}\;\S\;\CLG{C_6}^R})$
                in {\sf P(1,2,5,6)}, and
            \item $t_3 := (\purple{\CLG{C_2}\;\S\;\CLG{C_2}^R},\purple{\CNG{C_3}\;\S\;\CLG{C_3}^R},$\\
                \mbox{}$\qquad \quad\! \purple{\CNG{C_4}\;\S\;\CLG{C_4}^R},\purple{\CLG{C_5}\;\S\;\CLG{C_5}^R})$
                in {\sf P(2,3,4,5)}.
        \end{itemize}

        By Lemma~\ref{lm:4:almostclique}, this implies that all $C_i \cup C_j$ form a $2k$-clique and thus $C_1 \cup \ldots \cup C_6$ forms a $6k$-clique (see Figure~\ref{fg:decomposealmostfourclique} to check that all pairs are covered).

        For the other direction, consider a graph $G$ that contains a $6k$-clique $C^*$. Then we can split $C^*$ into 6
        vertex-disjoint $k$-cliques $C_1, \ldots, C_6$. Further we know that every three of these
        six $k$-cliques together form a $3k$-clique. Thus, the program {\sf C} generates the tuples $t_1,t_2,t_3$ as above.
        We can then use the three programs ${\sf P(\cdot,\cdot,\cdot,\cdot)}$ to generate such tuples
        surrounded with symbols $|$, $l_i$, and $r_i$ at appropriate positions. Adding the surrounding strings by ${\sf A}(T)$ and following the global structure of $\Gamma$
        generates the encoding $\GG{G}$.

        To see that $\Gamma$ is of constant size, note that we only use constantly many programs.
        Thus using a new set of terminal symbols for every instance of a program will still yield
        a constant total number of non-terminal symbols. Further, we only use 19 terminal symbols.
    \end{proof}

    The above lemma and the bound $|\GG{G}| = O(n^{k+1} \log n)$ imply the main theorem.
    \bibliographystyle{abbrv}
    \bibliography{tag}

\begin{thebibliography}{10}

\bibitem{ABV15}
A.~Abboud, A.~Backurs, and V.~V. Williams.
\newblock If the current clique algorithms are optimal, so is {V}aliant's
  parser.
\newblock In {\em 56th Annual Symposium on Foundations of Computer Science,
  FOCS'15}, pages 98--117, 2015.

\bibitem{abeille1988parsing}
A.~Abeill{\'e}.
\newblock Parsing french with tree adjoining grammar: some linguistic accounts.
\newblock In {\em 12th Conference on Computational Linguistics, COLING'88},
  pages 7--12, 1988.

\bibitem{backurs_et_al}
A.~Backurs, N.~Dikkala, and C.~Tzamos.
\newblock Tight hardness results for maximum weight rectangles.
\newblock In {\em 43rd International Colloquium on Automata, Languages, and
  Programming, ICALP'16}, volume~55, pages 81:1--81:13, 2016.

\bibitem{backurs2016improving}
A.~Backurs and C.~Tzamos.
\newblock Improving {V}iterbi is hard: Better runtimes imply faster clique
  algorithms.
\newblock In {\em 34th International Conference on Machine Learning, ICML'17},
  pages 311--321, 2017.

\bibitem{bringmann2016dichotomy}
K.~Bringmann, A.~Gr{\o}nlund, and K.~G. Larsen.
\newblock A dichotomy for regular expression membership testing.
\newblock In {\em 58th Annual IEEE Symposium on Foundations of Computer
  Science, FOCS'17}, pages 307--318, 2017.

\bibitem{DBLP:conf/cpm/Chang16}
Y.~Chang.
\newblock Hardness of {RNA} folding problem with four symbols.
\newblock In {\em 27th Annual Symposium on Combinatorial Pattern Matching,
  CPM'16}, pages 13:1--13:12, 2016.

\bibitem{DBLP:conf/tag/2016}
D.~Chiang and A.~Koller, editors.
\newblock {\em Proc.\ 12th International Workshop on Tree Adjoining Grammars
  and Related Formalisms (TAG+12), June 29 - July 1, 2016, Heinrich Heine
  University, D{\"{u}}sseldorf, Germany}. The Association for Computer
  Linguistics, 2016.

\bibitem{cocke}
J.~Cocke and J.~T. Schwartz.
\newblock Programming languages and their compilers: Preliminary notes.
\newblock Technical report, CIMS, NYU, 1970.

\bibitem{demberg2013incremental}
V.~Demberg, F.~Keller, and A.~Koller.
\newblock Incremental, predictive parsing with psycholinguistically motivated
  tree-adjoining grammar.
\newblock {\em Computational Linguistics}, 39(4):1025--1066, 2013.

\bibitem{doran1994xtag}
C.~Doran, D.~Egedi, B.~A. Hockey, B.~Srinivas, and M.~Zaidel.
\newblock {XTAG} system: a wide coverage grammar for {E}nglish.
\newblock In {\em 15th Conference on Computational Linguistics, COLING'94},
  pages 922--928, 1994.

\bibitem{earley1970efficient}
J.~Earley.
\newblock An efficient context-free parsing algorithm.
\newblock {\em C. ACM}, 13(2):94--102, 1970.

\bibitem{eisenbrand2004complexity}
F.~Eisenbrand and F.~Grandoni.
\newblock On the complexity of fixed parameter clique and dominating set.
\newblock {\em Theoretical Computer Science}, 326(1-3):57--67, 2004.

\bibitem{forbes2003d}
K.~Forbes, E.~Miltsakaki, R.~Prasad, A.~Sarkar, A.~Joshi, and B.~Webber.
\newblock {D-LTAG} system: Discourse parsing with a lexicalized tree-adjoining
  grammar.
\newblock {\em Journal of Logic, Language and Information}, 12(3):261--279,
  2003.

\bibitem{joshi1985tree}
A.~K. Joshi.
\newblock Tree adjoining grammars: How much context-sensitivity is required to
  provide reasonable structural descriptions?
\newblock In {\em Natural Language Processing, Theoretical, Computational and
  Psychological Perspectives}. Cambridge University Press, 1985.

\bibitem{joshi1975tree}
A.~K. Joshi, L.~S. Levy, and M.~Takahashi.
\newblock Tree adjunct grammars.
\newblock {\em JCSS}, 10(1):136--163, 1975.

\bibitem{joshi1997tree}
A.~K. Joshi and Y.~Schabes.
\newblock Tree-adjoining grammars.
\newblock In {\em Handbook of Formal Languages}, pages 69--123. Springer, 1997.

\bibitem{kasami}
T.~Kasami.
\newblock An efficient recognition and syntax algorithm for context-free
  languages.
\newblock Technical report, AFCRL-65-758, Air Force Cambridge Research Lab,
  Bedford, MA., 1965.

\bibitem{lee2002fast}
L.~Lee.
\newblock Fast context-free grammar parsing requires fast boolean matrix
  multiplication.
\newblock {\em JACM}, 49(1):1--15, 2002.

\bibitem{nevsetvril1985complexity}
J.~Ne{\v{s}}et{\v{r}}il and S.~Poljak.
\newblock On the complexity of the subgraph problem.
\newblock {\em Commentationes Mathematicae Universitatis Carolinae},
  26(2):415--419, 1985.

\bibitem{rajasekaran1998tal}
S.~Rajasekaran and S.~Yooseph.
\newblock {TAL} recognition in {$O(M(N^2))$} time.
\newblock {\em JCSS}, 56(1):83--89, 1998.

\bibitem{resnik1992probabilistic}
P.~Resnik.
\newblock Probabilistic tree-adjoining grammar as a framework for statistical
  natural language processing.
\newblock In {\em 14th Conference on Computational Linguistics, COLING'92},
  pages 418--424, 1992.

\bibitem{S94}
G.~Satta.
\newblock {Tree-adjoining Grammar Parsing and Boolean Matrix Multiplication}.
\newblock {\em Comput. Linguist.}, 20(2):173--191, June 1994.

\bibitem{schabes1988earley}
Y.~Schabes and A.~K. Joshi.
\newblock An {E}arley-type parsing algorithm for tree adjoining grammars.
\newblock In {\em 26th Annual Meeting of the Association for Computational
  Linguistics, ACL'88}, pages 258--269, 1988.

\bibitem{shieber1990synchronous}
S.~M. Shieber and Y.~Schabes.
\newblock Synchronous tree-adjoining grammars.
\newblock In {\em 13th Conference on Computational Linguistics, COLING'90},
  pages 253--258, 1990.

\bibitem{stone1997sentence}
M.~Stone and C.~Doran.
\newblock Sentence planning as description using tree adjoining grammar.
\newblock In {\em 35th Annual Meeting of the Association for Computational
  Linguistics, ACL'97}, pages 198--205, 1997.

\bibitem{uemura1999tree}
Y.~Uemura, A.~Hasegawa, S.~Kobayashi, and T.~Yokomori.
\newblock Tree adjoining grammars for {RNA} structure prediction.
\newblock {\em Theoretical Computer Science}, 210(2):277--303, 1999.

\bibitem{valiant1975general}
L.~G. Valiant.
\newblock General context-free recognition in less than cubic time.
\newblock {\em JCSS}, 10(2):308--315, 1975.

\bibitem{vassilevska2009efficient}
V.~Vassilevska.
\newblock Efficient algorithms for clique problems.
\newblock {\em Information Processing Letters}, 109(4):254--257, 2009.

\bibitem{VSJ85}
K.~Vijay-Shankar and A.~K. Joshi.
\newblock Some computational properties of tree adjoining grammars.
\newblock In {\em 23rd Annual Meeting of the Association for Computational
  Linguistics, ACL'85}, pages 82--93, 1985.

\bibitem{vijay1988feature}
K.~Vijay-Shanker and A.~K. Joshi.
\newblock Feature structures based tree adjoining grammars.
\newblock In {\em 12th Conference on Computational Linguistics, COLING'88},
  pages 714--719, 1988.

\bibitem{younger1967recognition}
D.~H. Younger.
\newblock Recognition and parsing of context-free languages in time $n^3$.
\newblock {\em Information and Control}, 10(2):189--208, 1967.

\end{thebibliography}

    \appendix

    \begin{sidewaysfigure}[ht]
        \begin{center}
            \scalebox{0.925}{
                \begin{tikzpicture}
                    \Tree [
                    .S [
                    .\node[color=blue](1){\udash{${\sf P(1,3,4,6)}_{In}$}};
                    \edge[color=blue, decorate, decoration={snake, segment length=6.1pt}]; [
                    .\node[color=blue](2){\udash{${\sf P(1,3,4,6)}_{Out}$}}; [
                    .\node[color=red](3){\uline{${\sf P(1,2,5,6)}_{In}$}};
                    \edge[color=red, decorate, decoration={snake, segment length=6.1pt}]; [
                    .\node[color=red](4){\uline{${\sf P(1,2,5,6)}_{Out}$}};
                    \edge[color=red, decorate, decoration={snake, segment length=6.1pt}]; [
                    .\node[color=red](5){\uline{${\sf P(1,2,5,6)}_{In}$}}; [
                    .\node[color=purple](6){\udot{${\sf P(2,3,4,5)}_{In}$}};
                    \edge[color=purple, decorate, decoration={snake, segment length=6.1pt}];[
                    .\node[color=purple](7){\udot{${\sf P(2,3,4,5)}_{Out}$}}; \edge[color=purple, decorate, decoration={snake, segment length=6.1pt}]; [
                    .\node[color=purple](8){\udot{${\sf P(2,3,4,5)}_{In}$}}; [
                    .\node[color=blue](9){\udash{${\sf P(1,3,4,6)}_{Out}$}}; \edge[color=blue, decorate, decoration={snake, segment length=6.1pt}];[
                    .\node[color=blue](10){\udash{${\sf P(1,3,4,6)}_{In}$}}; [
                    .\node(11) at (0,0) {$e$};
                    ]
                    ]
                    ]
                    ]
                    ]
                    ]
                    ]
                    ]
                    ]
                    ]
                    ]
                    ]
                    \node(12) at (-1.45,-11.49) {$r_3{\color{blue}\udash{$(C_3)$}}|\ldots $};
                    \draw[color=blue, dashed] (12.north east) to (10);
                    \draw[color=blue, dashed] (12.north west) to (9.south west);
                    \draw[color=blue, dashed] (12.north west) to (12.north east);

                    \node(13) at (-3.4,-11.49) {$\ldots |\,{\color{purple}\udot{$(C_3)$}}l_3$};
                    \draw[color=purple, dotted, thick] (13.north east) to (8.south west);
                    \draw[color=purple, dotted, thick] (13.north west) to (7.south west);
                    \draw[color=purple, dotted, thick] (13.north west) to (13.north east);

                    \node(14) at (-5.3,-11.49) {$r_2{\color{purple}\udot{$(C_2)$}}|\ldots $};
                    \draw[color=purple, dotted, thick] (14.north west) to (6.south west);
                    \draw[color=purple, dotted, thick] (14.north west) to (14.north east);

                    \node(15) at (-7.2,-11.49) {$\ldots |\,{\color{red}\uline{$(C_2)$}}l_2$};
                    \draw[color=red] (15.north east) to (5.south west);
                    \draw[color=red] (15.north west) to (4.south west);
                    \draw[color=red] (15.north west) to (15.north east);

                    \node(16) at (-9.1,-11.49) {$r_1{\color{red}\uline{$(C_1)$}}\;|\ldots $};
                    \draw[color=red] (16.north west) to (3.south west);
                    \draw[color=red] (16.north west) to (16.north east);

                    \node(17) at (-11.05,-11.49) {$\ldots |{\color{blue}\udash{$(C_1)$}}l_1$};
                    \draw[color=blue, dashed] (17.north east) to (2.south west);
                    \draw[color=blue, dashed] (17.north west) to (1.south west);
                    \draw[color=blue, dashed] (17.north west) to (17.north east);

                    \node(18) at (1.45,-11.49) {$\ldots |\,{\color{blue}\udash{$(C_4)$}}l_4$};
                    \draw[color=blue, dashed] (18.north west) to (10);
                    \draw[color=blue, dashed] (18.north east) to (9.south east);
                    \draw[color=blue, dashed] (18.north east) to (18.north west);

                    \node(19) at (3.4,-11.49) {$r_4{\color{purple}\udot{$(C_4)$}}\,|\ldots $};
                    \draw[color=purple, dotted, thick] (19.north west) to (8.south east);
                    \draw[color=purple, dotted, thick] (19.north east) to (7.south east);
                    \draw[color=purple, dotted, thick] (19.north east) to (19.north west);

                    \node(20) at (5.2,-11.49) {$\ldots |{\color{purple}\udot{$(C_5)$}}l_5$};
                    \draw[color=purple, dotted, thick] (20.north east) to (6.south east);
                    \draw[color=purple, dotted, thick] (20.north east) to (20.north west);

                    \node(21) at (7.1,-11.49) {$r_5{\color{red}\uline{$(C_5)$}}\,|\ldots $};
                    \draw[color=red] (21.north west) to (5.south east);
                    \draw[color=red] (21.north east) to (4.south east);
                    \draw[color=red] (21.north east) to (21.north west);

                    \node(22) at (9,-11.49) {$\ldots |\,{\color{red}\uline{$(C_6)$}}l_5$};
                    \draw[color=red] (22.north east) to (3.south east);
                    \draw[color=red] (22.north east) to (22.north west);

                    \node(23) at (10.89,-11.49) {$r_6{\color{blue}\udash{$(C_6)$}}|\ldots $};
                    \draw[color=blue, dashed] (23.north west) to (2.south east);
                    \draw[color=blue, dashed] (23.north east) to (1.south east);
                    \draw[color=blue, dashed] (23.north east) to (23.north west);
            \end{tikzpicture}}
        \end{center}
        \caption{\emph{Enlarged version of Figure~\ref{fg:gamma}.} Global structure of a parsing of $\GG{G}$ by $\Gamma$.
        (Clique gadgets are abbreviated.)}\label{fg:gammamaksdn}
    \end{sidewaysfigure}
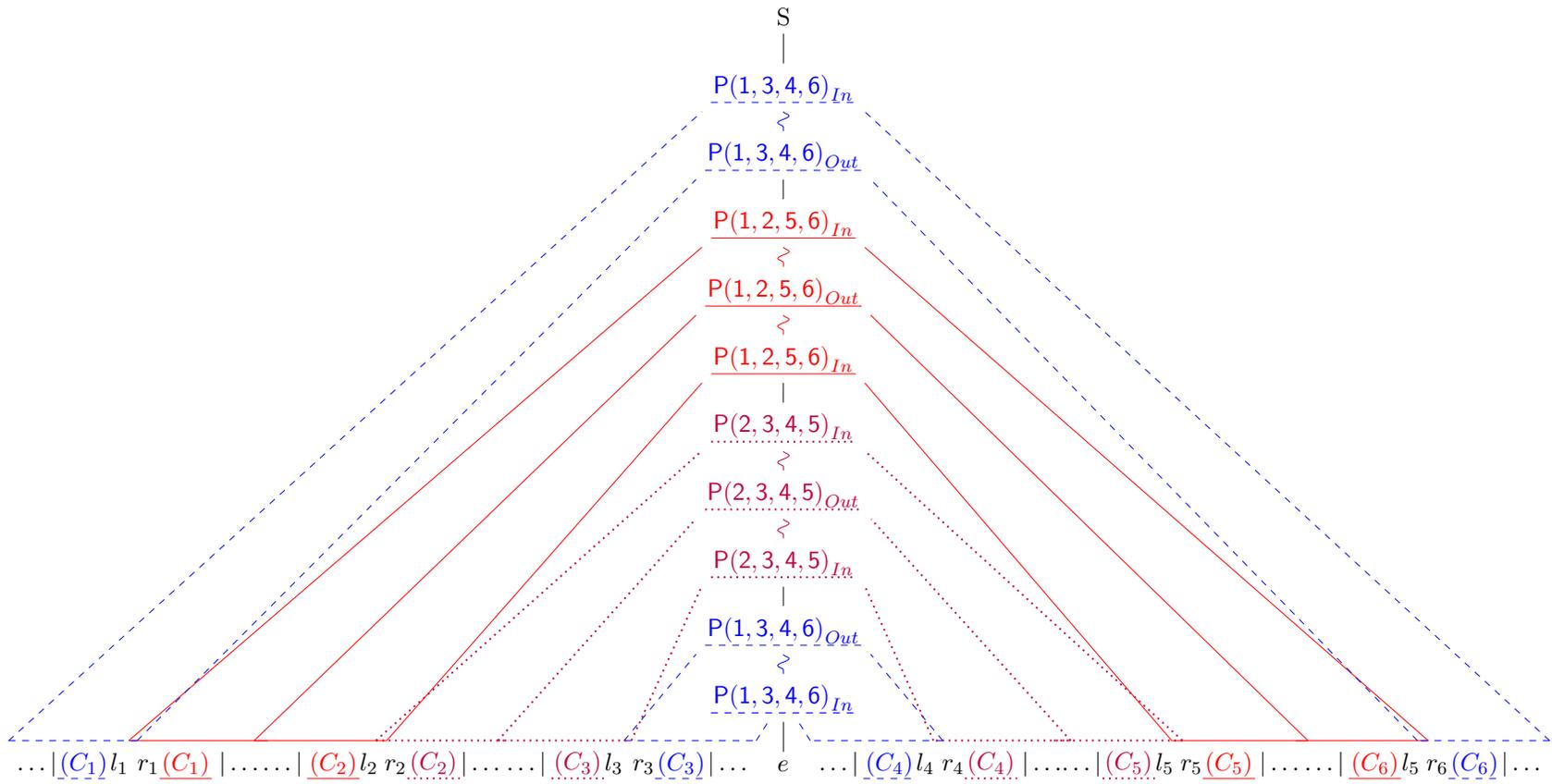

    \end{document}